\documentclass{article}

\usepackage{graphicx}
\usepackage{amsthm}

\newcommand{\beq}{\begin{eqnarray}}
\newcommand{\eeq}{\end{eqnarray}}
\newcommand{\la}{\langle}
\newcommand{\ra}{\rangle}
\newcommand{\gamm}{\textrm{\raisebox{.4ex}{$\gamma$}}}

\newtheorem{thm}{Theorem}[section]
\newtheorem{lem}{Lemma}[section]
\newtheorem{prop}{Proposition}[section]
\newtheorem{cor}{Corollary}[section]
\newtheorem{defn}{Definition}[section]
\newtheorem*{rem}{Remark}

\begin{document}
\author{K.J. Sharkey, I.Z. Kiss, R.R. Wilkinson, P.L. Simon} 
\title{Exact equations for SIR epidemics on tree graphs} 
\date{}
\maketitle

\section*{Abstract}
We consider Markovian susceptible-infectious-removed (SIR) dynamics on time-invariant weighted contact networks where the infection and removal processes are Poisson and where network links may be directed or undirected. We prove that a particular pair-based moment closure representation  generates the expected infectious time series for networks with no cycles in the underlying graph. Moreover, this ``deterministic'' representation of the expected behaviour of a complex heterogeneous and finite Markovian system is straightforward to evaluate numerically.

\section{Introduction}
\label{s1}
\subsection{Background}
The majority of epidemic models fall either into the category of stochastic models (Bailey 1975; Bartlett 1956) or into the category of deterministic differential equation-based models (Anderson and May 1991; Kermack and McKendrick 1927). These two strands developed largely independently for much of the twentieth century. Thus, an interesting question arises as to the precise mathematical connection between stochastic and deterministic models. Frequently, deterministic descriptions apply to large populations where the stochastic effects can be treated as negligible. For small populations we shall assume that it is the average or expected behaviour of the epidemic that we are hoping to replicate with ``deterministic'' descriptions. This average behaviour is a system characteristic that is fully specified by the system and its initial conditions. 

The first epidemic models were based on the assumption that populations are evenly mixed, with each individual equally likely to interact with any other individual at any time (Heathcote 2000). A classic example of this type of model is the Susceptible-Infectious-Removed (SIR) compartmental model whereby individuals are classified according to being in one of these three states. It has been shown that for this type of mean-field model, the average of many stochastic simulations (the expected outcome of the stochastic model) converges to the solution of the ``equivalent'' mean-field deterministic model in the limit of an infinite population size and subject to strict conditions regarding the initialisation of the epidemic (Kurtz 1970, 1971; Simon and Kiss 2011).

More recently, a higher degree of realism has been introduced by considering stochastic models on contact networks where individuals are only able to contact a limited subset of the population. This enables significant heterogeneity to be incorporated, treating individuals as distinct entities with fixed connectivity to pre-allocated neighbours. While stochastic models are readily extended to incorporate such systems, deterministic descriptions have been more problematic. Several methodologies have been developed including pair-approximation models (Keeling 1999; Keeling and Eames 2005; Rand 1999), degree-based models (Pastor-Satorras and Vespignani 2001), and models based on the probability generating function (PGF) formalism which are applicable to configuration networks (Volz 2008) as well as the related edge-based compartmental modelling (Miller et al. 2012; Miller and Volz 2012). It has been observed (House and Keeling 2011) that these models are, at some level, equivalent and are all derived from similar principles of independence. Although comparison with simulation of stochastic models can sometimes be good, the basic link remains obscure. 

Typically there are two idealised scenarios in which exact correspondence between stochastic models and solvable deterministic descriptions has been shown. Firstly, correspondence has been shown to sometimes occur in the limit of infinite populations for particular idealised graphs (Ball and Neal 2008; Decreusefond et al. 2012) which cannot be exactly realised in practice. It can also occur with some very simplified systems whose symmetry properties can be exploited to achieve reductions in the stochastic description (Keeling and Ross 2008; Simon et al. 2011). 

Here we consider a recently introduced class of model, related to the pair-approximation models, which give an exact correspondence between a deterministic description and the stochastic model for SIR epidemics on finite, time-invariant networks. Pair-approximation models were introduced into  network-based epidemic and ecological theory in the 1990s to describe large populations of interacting individuals (Matsuda et al. 1992; Sato et al. 1994; Harada and Iwasa 1994; Rand 1999; Keeling 1999). They are an example of a hierarchy of equations which are truncated at the second order by an approximation (truncation at the first order corresponds to mean-field). This type of hierarchy was first considered in statistical physics and is sometimes known as the Bogoliubov-Born-Green-Kirkwood-Yvon (BBGKY) hierarchy (Kirkwood 1946, 1947; Born and Green 1946). Recently, related models have been considered at the level of individuals, variously called subsystem equations, moment dynamics equations, pair-based equations (Sharkey 2008, 2011; Baker and Simpson 2010; Markham et al. 2013). This method generates a solvable class of models which can encompass a significant amount of heterogeneity and enables a fundamental link with finite stochastic models (Sharkey 2008, 2011). 

We consider a pair-based representation of Markovian SIR dynamics. We show that by considering subsystems at the level of pairs, a closure can be found that determines the expected infectious time series exactly for arbitrary network structures where the underlying graph is a tree and, in some special circumstances, for particular networks with cycles. We note that the recent, related message passing formulation of epidemics on contact networks developed by Karrer and Newman (2010) also enables an exact description of epidemic dynamics on finite tree graphs.

\subsection{Statement of the main result} \label{s2}

We consider an SIR compartmental model composed of $P$ individuals
whose states are described at any given point in time by vectors $I$ and $S$ with respective components $I_i$ and $S_i$, $i\in \{1,2,...,P\}$
such that $I_i=1$ if individual $i$ is infectious ($I_i=0$ otherwise) and $S_i=1$ if individual
$i$ is susceptible ($S_i=0$ otherwise). Transmission and recovery occur by Poisson processes with
rate parameters $\lambda_i=\sum_{j=1}^PT_{ij}I_jS_i$ and $\mu_i=\gamma_iI_i$, respectively 
where $T$ is a ``transmission'' matrix with (time-independent) elements $T_{ij}$ denoting the rate
parameter for an infectious node $j$ infecting a susceptible node $i$ ($T_{ii}=0$ for all $i$) and where $\gamma_i$ denotes the rate parameter for an infectious individual $i$ to recover, enabling individual-specific removal rates.

As shown by Sharkey (2011), for any transmission matrix $T$ and any nodes $i$, $j$
the following differential equations are provably exact (consistent with the stochastic model):
\begin{eqnarray}\nonumber
\dot{\la S_i\ra}&=&-\sum_{j=1}^P T_{ij}\la S_iI_j\ra, \\ \nonumber \dot{\la I_i\ra}&=&\sum_{j=1}^P
T_{ij}\la S_iI_j\ra -\gamma_i\la I_i\ra,\\ \nonumber
\dot{\la S_iI_j\ra}&=&\sum_{k=1,k\neq i}^P T_{jk}\la S_iS_jI_k\ra-\sum_{k=1, k\neq j}^P T_{ik}\la I_kS_iI_j\ra,    \\  \nonumber
&& -T_{ij}\la S_iI_j\ra-\gamma_j\la S_iI_j\ra, \\  
\dot{\la S_iS_j\ra}&=&-\sum_{k=1, k\neq j}^PT_{ik}\la I_kS_iS_j\ra-\sum_{k=1, k\neq i}^PT_{jk}\la
S_iS_jI_k\ra, \label{0.3}
\end{eqnarray}
where $\la S_i\ra$ and $\la I_i\ra$ denote the time-dependent probabilities (or equivalently the expected values of the indicator functions) for individual $i$ to
be susceptible and infectious, respectively, and expressions of the form $\la A_iB_j\ra$ denote the
time-dependent probability that individual $i$ is in state $A$ and individual $j$ is in state $B$ with a similar interpretation of terms of the form $\la A_iB_jC_k\ra$. Here and throughout, we adopt the dot notation to denote time derivatives. It follows that the expected population-level susceptible and infectious time series are given by $\sum_{i=1}^P\la S_i\ra$ and $\sum_{i=1}^P\la I_i\ra$ respectively.

Note that it is a short step (see Sharkey 2008) from (\ref{0.3}) to the familiar population-level pair equations (Keeling 1999; Keeling and Eames 2005; Rand 1999), also proved independently by Taylor et al. (2012) for the susceptible-infectious-susceptible variant.

This system can be completed by formulating differential equations for the triples, quadruples, and so forth until we reach the full system size. This yields a self-contained system of differential
equations that exactly determines the probabilities of each quantity given initial conditions. However, cascading these equations up to the full system size will usually result in a system that is impractical to solve due to its sheer size. This is why this system is typically closed at
some level by introducing a functional relation approximating higher-order probabilities in terms of lower-order ones. One of the most
frequently used closure relations can be written as \beq \la
A_iB_jC_k\ra\approx\frac{\la A_iB_j\ra\la B_jC_k\ra}{\la B_j\ra} \label{kw} \eeq for the current context. Applying this closure
relation to our system at the level of pairs we arrive at the following system:
\begin{eqnarray}\nonumber
\dot{\la X_i\ra}&=&-\sum_{j=1}^P T_{ij}\la X_iY_j\ra, \\ \nonumber \dot{\la Y_i\ra}&=&\sum_{j=1}^P
T_{ij}\la X_iY_j\ra -\gamma_i\la Y_i\ra,\\ \nonumber
\dot{\la X_iY_j\ra}&=&\sum_{k=1,k\neq i}^P T_{jk}\frac{\la X_iX_j\ra\la X_jY_k\ra}{\la X_j\ra}-\sum_{k=1,k\neq j}^P T_{ik}\frac{\la X_iY_k\ra\la X_iY_j\ra}{\la X_i\ra},    \\ \nonumber
&&-T_{ij}\la X_iY_j\ra-\gamma_j\la X_iY_j\ra, \\ \nonumber
\dot{\la X_iX_j\ra}&=&-\sum_{k=1,k\neq j}^PT_{ik}\frac{\la Y_kX_i\ra\la X_iX_j\ra}{\la
X_i\ra}-\sum_{k=1,k\neq i}^PT_{jk}\frac{\la X_iX_j\ra\la X_jY_k\ra}{\la X_j\ra}, \\ \label{0.2}
\end{eqnarray}
where we use $X$ for susceptible and $Y$ for infectious to emphasise that these are approximating
differential equations based on the closure. When $\la X_i\ra$ in the denominator is zero, we assume that the approximation takes the value zero.

In general, we consider networks (graphs) with directed and undirected edges. In what follows, we use the terminology ``tree graph'' to include graphs with directed edges where the underlying (equivalent undirected) graph is a tree. Our main aim is to show that when matrix $T$ represents a tree and the
system is initiated in a pure system state (that is, one of the $3^P$ possible configurations has
probability 1 at time $t=0$), then the system can be closed at the level of pairs such that the closure holds exactly.   Specifically, solving the closed system above, we obtain the same values for all marginal and pairwise-joint probabilities present in the unclosed system: $\la X_i\ra=\la S_i\ra$, $\la Y_i\ra=\la I_i\ra$, with similar equalities holding for the pairs.

In fact, we will prove the following theorem.

\begin{thm}
\noindent Let us assume the following:
\begin{itemize}
\item The graph (transmission network) is a tree (the underlying graph has no cycles).
\item The initial condition is a pure state, i.e. the system is initially in one of its $3^P$ possible configurations with probability 1.
\end{itemize}
Then the following relations hold: \beq \la S_j\ra\la S_iS_jI_k\ra=\la S_iS_j\ra\la S_jI_k\ra \nonumber\label{theoremMain_1}\eeq
for all $i\in\{1,2,...,P\}$ and for all $j$ with links towards $i$ and all $k$ with links towards $j$ : $i\neq k$;
\beq \la S_i\ra\la I_kS_iI_j\ra=\la I_kS_i\ra\la S_i I_j\ra \nonumber\label{theoremMain_2}\eeq for all $i\in\{1,2,...,P\}$
and for all $j$ and $k$ with links towards $i$: $j\neq k$. 
\label{theoremMain}
\end{thm}
\begin{rem}
This theorem also holds for mixed (probabilistic) initial system states provided that the initial probabilities of the states of individuals in the system are uncorrelated. However, in general, mixed initial states cannot be represented exactly. 
\end{rem}
The theorem will be formulated in a more general context stating that even higher-order closure
relations are also exact. 

Figure~\ref{example}
\begin{figure}[htb]
    \centerline{\includegraphics[width=1\textwidth]{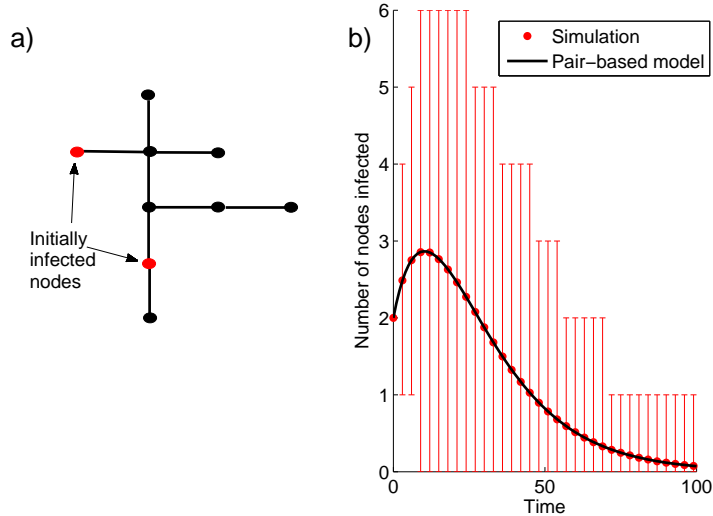}}
    \caption{a) An undirected tree indicating two nodes which we  infect to initiate epidemics, with all other nodes initially susceptible. b) The mean (dots) of 100,000 stochastic simulations on the network with transmission rate $\tau=0.1$ across each link and removal rate $\gamma=0.05$ for each node, with error bars denoting the 5th and 95th percentiles plotted together with the solution of (\ref{0.2}) (solid line) using the Matlab code published with Sharkey (2011)}
    \label{example}
\end{figure}
shows the numerical solution of (\ref{0.2}) for a small network of 9 nodes where it is clear that it is accurate to within the precision visible on the graph. Matlab code for solving the system of equations~(\ref{0.2}) is provided in Sharkey (2011). This code also works on networks which are not trees but is no longer exact in these cases. Cycles in the underlying graph of order three utilise the alternative closure $\la A_iB_jC_k\ra=\la A_iB_j\ra\la B_jC_k\ra\la
A_iC_k\ra/\la A_i\ra\la B_j\ra\la C_k\ra$ which is believed to gain increased accuracy in most circumstances, but these do not occur in the tree graphs considered in the present work.

The structure of the paper is as follows. Section~\ref{s0.01} introduces some notation which is needed to prove the result. This also contains an important theorem (Theorem~\ref{theorem1}) which specifies equations describing the probabilities of the states of arbitrary subsystems (the proof of this result is given in Appendix A). The relevant state space for our domain of a tree graph is then developed. Section 3 proves the main result, initially focusing on some special cases to help motivate and facilitate understanding of the main ideas and steps of the general proof in Section~\ref{sec3.5}. The main ingredient for the general proof is Lemma~\ref{lemma3} which is proved via Theorem~\ref{theorem1}. Theorem~\ref{theorem2} then follows easily by induction from Lemma~\ref{lemma3}. The theorem as stated above is a simple corollary of Theorem~\ref{theorem2}. In Section~\ref{s4.0} we discuss an application of the pair-based model to some special cases of graphs with cycles where it is also exact.

\section{Formulating the full system}
\label{s0.01}
In this section we introduce a new notation which will assist in formulating the set of
differential equations for the full system. In (\ref{0.3}) we formulated the differential
equations up to the level of pairs and we said that this could be continued up to the full system
level. This will be done formally here. In order to make the method clearer, using our existing notation let us first evaluate the full set of equations for the undirected line graph with three nodes which we refer to as the open triple, depicted in Figure~\ref{figm2}.
\begin{figure}[htb]
    \centerline{\includegraphics[width=0.3\textwidth]{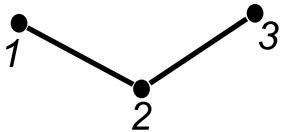}}
    \caption{Open triple graph}
    \label{figm2}
\end{figure}
Here we shall assume that the transmission rate parameter is $\tau$ across both links and that the removal rate is $\gamma$ for all three nodes. Firstly we write all of the single node equations
for this network. From (\ref{0.3}):
\begin{eqnarray}\nonumber
\dot{\la I_1\ra}&=&\tau\la S_1I_2\ra-\gamma\la I_1\ra, \\ \nonumber \dot{\la I_2\ra}&=&\tau\la I_1S_2\ra +\tau\la S_2I_3\ra -\gamma\la I_2\ra,\\
\dot{\la I_3\ra}&=&\tau\la I_2S_3\ra-\gamma\la I_3\ra, 
 \label{3.14}
\end{eqnarray}
and
\begin{eqnarray}\nonumber
\dot{\la S_1\ra}&=&-\tau\la
S_1I_2\ra, \\ \nonumber \dot{\la S_2\ra}&=&-\tau\la I_1S_2\ra-\tau\la S_2I_3\ra, \\ 
\dot{\la S_3\ra}&=&-\tau\la I_2S_3\ra. \label{3.141}
\end{eqnarray}
We also need to specify the following equations for pairs:
\begin{eqnarray}\nonumber
\dot{\la S_1I_2\ra}&=&\tau\la S_1S_2I_3\ra - \tau\la S_1I_2\ra-\gamma\la S_1I_2\ra, \\ \nonumber
\dot{\la I_1S_2\ra}&=&-\tau\la I_1S_2I_3\ra-\tau\la I_1S_2\ra-\gamma\la I_1S_2\ra, \\ \nonumber
\dot{\la S_2I_3\ra}&=&-\tau\la I_1S_2I_3\ra -\tau\la S_2I_3\ra-\gamma\la S_2I_3\ra, \\
\dot{\la I_2S_3\ra}&=&\tau \la I_1S_2S_3\ra-\tau\la I_2S_3\ra-\gamma\la I_2S_3\ra. \label{3.15}
\end{eqnarray}
Finally, at the triple level we have from the master equation (since the system has only three
nodes):
\begin{eqnarray}\nonumber
\dot{\la S_1S_2I_3\ra}&=&-\tau\la S_1S_2I_3\ra-\gamma\la S_1S_2I_3\ra, \\ \nonumber
\dot{\la I_1S_2I_3\ra}&=&-2\tau\la I_1S_2I_3\ra-2\gamma\la I_1S_2I_3\ra, \\
\dot{\la I_1S_2S_3\ra}&=&-\tau\la I_1S_2S_3\ra - \gamma\la I_1S_2S_3\ra. \label{3.16}
\end{eqnarray}

In order to formulate the full system for an arbitrary graph, we introduce notation for the
subsystem states.

\subsection{Notation for system and subsystem states} \label{s3}

In general, our stochastic system (which we denote by $\Gamma$) comprises of $P$ individuals, each
of which may be in any of the $S$, $I$ or $R$ states at any given time. In total, this corresponds
to $3^P$ possible states. Denoting these system states by $\Gamma^\alpha$,
$\alpha\in\{1,2,...,3^P\}$, the probabilities for each state are given by the master equation (or Kolmogorov equations): \beq
\dot{\la\Gamma^\alpha\ra}=\sum_{\beta=1}^{3^P}\sigma^{\alpha\beta}\la\Gamma^\beta\ra-\sum_{\beta=1}^{3^P}\sigma^{\beta\alpha}\la\Gamma^\alpha\ra,
\label{eq3.1} \eeq where $\sigma$ denotes a constant matrix of Poisson rate parameters. The master
equation completely describes our stochastic system using a set of $3^P$ ordinary differential
equations. Our overall objective is to show that (\ref{0.2}) is implied by the master
equation when $T$ represents a tree graph.

It is useful for us to define a general subsystem $\psi_W$ comprising of $r$ nodes in $\Gamma$
indexed by vector $W$ of length $r$: $W=(W_1,W_2,...,W_i,...,W_r)$, $W_i\in \{1,2,...,P\}$ where we can
assume $W_1<W_2<...<W_r$. We assume that the network connections of the nodes of $\psi_W$ are a
subset of the connections of $\Gamma$.

Let $\psi_W^A$ denote the state of subsystem $\psi_W$ where $A=(A_1,A_2,...,A_r)$ and
$A_i\in\{S,I,R\}$ $\forall i\in\{1,2,...,r\}$ is a sequence of $S,I$ and $R$ symbols of length $r$
such that the state of node $W_i$ is $A_i$. In terms of the notation of the previous section for
subsystems of single nodes and pairs of nodes, we have $S_i=\psi_i^S$, $I_i=\psi_i^I$,
$S_iI_j=\psi_{i,j}^{SI}$, $S_iS_jI_k=\psi_{i,j,k}^{SSI}$ etc. We shall use these two notations
interchangeably. We shall also sometimes treat indexing vectors such as $W$ as sets such that $n\in W$ means that the
node $n$ is in the subsystem $\psi_W$.

In general, although we can specify the states of each node with this type of notation, an important ambiguity remains because information about the network structure is not included. To remove
this ambiguity, the notation should normally be used in the context of a sketch of the relevant
network structure or where the network structure is clear from the context of its use (as in
(\ref{0.3})).

Let us now show how the differential equations of the different subsystem states can be formulated in general.

\subsection{Differential equations for subsystems} \label{s7}

Here we obtain differential equations describing the rate of change of the state of any subsystem. First we make some definitions.
\begin{defn}
A neighbour of node $i$ is a node with a network link directed towards $i$.
\end{defn}
\begin{defn}
$N_i$ denotes the set of neighbours of node $i$. That is: $T_{ij}\neq
0$ $\forall j\in N_i$.
\end{defn}

\begin{defn}
For the subsystem state $\psi_W^A$, if node $W_k$ is infectious then: \beq
h_{W_k}(\psi_W^A)=\psi_W^{A_1...A_{k-1}SA_{k+1}...A_r}. \nonumber\eeq Otherwise, $h_{W_k}(\psi_W^A)=\psi_W^A$.
\end{defn}
\begin{rem}
This operator changes the state of node $W_k$ in subsystem $\psi_W$ to $S$ if it is infectious. If
node $W_k$ is susceptible or removed then it leaves the state unchanged.
\end{rem}

\begin{defn}
For the subsystem state $\psi_W^A$ of $r$ nodes, a subsystem of $r+1$ nodes can be generated as
follows: Take $k\in\{1,2,...,r\}$ and take a neighbour $n$ of $W_k$ outside of the subsystem with
a network link towards $W_k$, i.e. let $n\in N_{W_k}$, $n\notin W$. If $A_k=S$, then the generated
subsystem state of $r+1$ nodes is given by the generating rule: \beq
g^n_{W_k}(\psi_W^A)=\psi^{A_1...A_rI}_{W_1,...,W_r,n},  \nonumber\eeq i.e. the subsystem is
extended by an infected at node $n$ which is connected towards $W_k$. If $A_k=I$, then the
generated subsystem state is given by: \beq
g^n_{W_k}(\psi_W^A)=\psi^{A_1...A_{k-1}SA_{k+1}...A_rI}_{W_1,...,W_r,n},  \nonumber\eeq i.e.
the subsystem is extended by an infected at node $n$ which is connected towards $W_k$ and the
state of node $W_k$ is changed from $I$ to $S$.

To complete the definition, if $A_k=R$ then the operator $g_{W_k}^n$ leaves the subsystem unchanged. We
also assume that for any state $A_k$ where there is no link from node $n$ to node $W_k$ in the transmission matrix $T$, then the subsystem is also left unchanged.
\label{GR}
\end{defn}
\begin{rem}
The generated order $r+1$ subsystem is obtained by replacing a susceptible or infectious node $W_k$
in the original subsystem by an $SI$ arc such that the $S$ node of the arc is put in the place
of the node $W_k$ and where the $I$ node of the arc is external to the subsystem.
\end{rem}
\begin{defn}
For the subsystem $\psi_W^A$, if node $W_k$ is removed then: \beq
f_{W_k}(\psi_W^A)=\psi_W^{A_1...A_{k-1}IA_{k+1}...A_r}. \nonumber\eeq Otherwise, $f_{W_k}(\psi_W^A)=\psi_W^A$.
\end{defn}

\begin{defn}
For any subsystem $\psi_W$ of $r$ nodes in state $\psi_W^A$ we define $D_k^{Aa}$ where $k\in\{1,2,...,r\}$ and $a\in\{S,I,R\}$ to have value 1 if $A_k=a$ and to have value zero otherwise: \beq D_k^{Aa}=\left\{
\begin{array}{ll}
1 & \textrm{if $A_k=a$,} \\
0 & \textrm{otherwise.}
\end{array}\right.
\nonumber \eeq 
\label{D_def}
\end{defn}

\begin{thm}
The rate of change of the probability of a subsystem state $\psi^A_W$ is:
\begin{eqnarray}\nonumber
\dot{\la\psi_W^A\ra}&=&\sum_{k=1}^r\left( 1-D_k^{AR}\right )\left [(-1)^{D_k^{AS}}\left (\sum_{n=1,n\notin W}^PT_{W_kn}\la g_{W_k}^n(\psi_W^A)\ra \right. \right. \\ \nonumber
&&+\left.\left. \sum_{l=1}^rT_{W_kW_l}D_l^{AI}\la h_{W_k}(\psi_W^A)\ra\right
)-D_k^{AI}\gamm_{W_k}\la\psi_W^A\ra\right ] \\ 
&&+\sum_{k=1}^rD_k^{AR}\gamm_{W_k}\la f_{W_k}(\psi_W^A)\ra . \label{14}
\end{eqnarray}
\label{theorem1}
\end{thm}
The proof of this theorem is a rather long diversion and can be found in Appendix A.

As an example of applying the theorem, we can use it to obtain the set of subsystem
equations~(\ref{0.3}) by considering each equation in turn:
\begin{itemize}
\item If the subsystem is a single susceptible individual $\psi_i^S$, then $r=1$ so $k$ can only take the value $k=1$ where $W_1=i$ and $A_1=S$, reducing (\ref{14}) to:
\beq \dot{\la\psi_i^S\ra}=-\sum_{n=1,n\neq i}^PT_{in}\la \psi_{i,n}^{SI}\ra. \nonumber\eeq The first term on
the second line of (\ref{14}) is zero because $T_{ii}=0$, and the other terms are zero because
$D_1^{SI}=0$ and $D_1^{SR}=0$.

\item For an infectious individual $\psi_i^I$ we obtain:
\beq \dot{\la\psi_i^I\ra}=\sum_{n=1,n\neq i}^PT_{in}\la \psi_{i,n}^{SI}\ra-\gamma_i\la\psi_i^I\ra,
\nonumber\eeq where the first term on the second line of (\ref{14}) is zero because $T_{ii}=0$ and the last term is zero because $D_1^{IR}=0$.
\item If the subsystem is the pair $\psi_{i,j}^{SI}$ then the sum over $k$ is over $k=1$ and $k=2$ and $W_1=i$, $W_2=j$, $A_1=S$, $A_2=I$ so:
\begin{eqnarray}\nonumber
\dot{\la \psi_{i,j}^{SI}\ra}&=&-\sum_{n=1,n\notin\{i,j\}}^PT_{in}\la\psi_{n,i,j}^{ISI}\ra -T_{ij}\la \psi_{i,j}^{SI}\ra \\ \nonumber
&&+\sum_{n=1,n\notin\{i,j\}}^PT_{jn}\la\psi_{i,j,n}^{SSI}\ra-\gamma_j\la\psi_{i,j}^{SI}\ra,
\end{eqnarray}
where the first line corresponds to $k=1$ and the second to $k=2$.
\item If the subsystem is the pair $\psi_{i,j}^{SS}$ then the sum is over $k=1$ and $k=2$ where $W_1=i$, $W_2=j$, $A_1=S$ and $A_2=S$ so:
\beq \dot{\la\psi_{i,j}^{SS}\ra}=-\sum_{n=1,n\notin \{i,j\}}^PT_{in}\la\psi_{n,i,j}^{ISS}\ra
-\sum_{n=1,n\notin\{i,j\}}^PT_{jn}\la\psi_{i,j,n}^{SSI}\ra,\nonumber \eeq where both terms come from the first
line of (\ref{14}).
\end{itemize}
We have therefore obtained (\ref{0.3}) in a slightly different notation (recall that $T_{ii}=0$ $\forall i\in\{1,2,...,P\}$).

\subsection{The state space for a tree graph} \label{s8}

Here we build up a state space which is sufficient to describe a tree graph. We first make
some definitions.
\begin{defn}
An $r$-motif is a subsystem of $\Gamma$ comprising of $r$ nodes and of network links such
that it forms a weakly connected network.
\end{defn}
\begin{defn}
An $r$-state is the state of an $r$-motif.
\end{defn}

The state space that we need to consider is built up inductively from the states of single nodes
by considering the infection process. Starting with the infected states of the single nodes
$\psi_i^I$, $i\in\{1,2,...,P\}$, (\ref{14}) shows that they depend on the 2-states
$\psi_{i,j}^{SI}$, $j\in N_i$ as described by the generating rule (Definition~\ref{GR}).

The differential equations for $\la\psi_{i,j}^{SI}\ra$ in turn contain the 3-states
$\psi_{i,j,k}^{SSI}$, $k\in N_j$ and $\psi_{k,i,j}^{ISI}$, $k\in N_i$. The differential equations
for the 3-states contain 4-states and typically, the differential equations for $r$-states contain
$(r+1)$-states for $r\in\{1,2,...,(P-1)\}$. This state generation process can continue until we
reach $P$-states which can only depend on other $P$-states.

Note that this process always forms subsystems which are motifs and that the motif states can
never include removed nodes.
\begin{defn}
An out-neighbour of node $i$ is a node with a network link from $i$ towards it.
\end{defn}
\begin{prop}
For a tree graph, if the out-neighbours of the $I$ nodes are all $S$ in an $r$-motif, then this is true for all $(r+1)$-motifs generated from this $r$-motif.
\label{prop7}
\end{prop}
\begin{proof}
This follows easily from the definition of the generating rule (Definition~\ref{GR}).
\end{proof}

\begin{defn}
Consider a tree graph and take the 1-motifs with $I$ nodes: $\psi_i^I, i\in\{ 1,2,...,P\}$.
The ``basic state space'' $M$ is formed by these 1-states together with the set of motif states that
can be iteratively generated from them using the generating rule (Definition~\ref{GR}).
\end{defn}
\begin{rem}
Due to the method of its construction, the state space $M$ gives a self-contained system of
differential equations, i.e. the time derivatives of the probabilities of each motif state can be
expressed in terms of the probabilities of other motif states in the state space. An example in the case of the open triple is given by the motifs in (\ref{3.14}), (\ref{3.15}) and~(\ref{3.16}). 
\end{rem}

\begin{defn}
Consider a tree graph and the 1-states: $\{\psi_i^I,\psi_i^S : i\in\{1,...,P\}\}$ and the
2-states with $SS$, i.e. $\{\psi_{i,j}^{SS} : i\in\{1,...,P\}, j\in N_i\}$. The ``extended state
space'' $\bar{M}$ comprises of these motif states together with the set of motif states that can be
generated from them by repeated iteration of the generating rule.
\end{defn}
\begin{rem}
The extended state space is required to form the relevant closure relations. Due to the method of its construction, it is also self-contained.
\end{rem}
\begin{lem}
Let $\psi_W^A\in\bar{M}$. Then the out-neighbour of an $I$ node is an $S$ node in $\psi_W^A$. \label{lemma2}
\end{lem}
\begin{proof}

Follows from Proposition~\ref{prop7}.
\end{proof}
\begin{lem}
For a tree graph, the equation for the time derivative of the probability of an $r$-state $\psi_W^A\in\bar{M}$ is given by:
\begin{eqnarray}\nonumber
\dot{\la\psi_W^A\ra}&=&\sum_{k=1}^r\left [(-1)^{D_k^{AS}}\sum_{n=1,n\notin W}^PT_{W_kn}\la g_{W_k}^n(\psi_W^A)\ra  \right. \\
&&-D_k^{AS}\left.
\sum_{l=1}^rT_{W_kW_l}D_l^{AI}\la\psi_W^A\ra-D_k^{AI}\gamm_{W_k}\la\psi_W^A\ra\right ]. \label{31}
\end{eqnarray}
\label{lemma1}
\end{lem}
\begin{proof}
For these states we have $D_k^{AR}=0$ $\forall k\in \{1,2,...,r\}$. Additionally, when
$D_k^{AI}=1$, the first term on the second line of (\ref{14}) never arises because
$D_l^{AI}=1$ implies that an $I$ is connected to an $I$ node in r-state $\psi_W^A$ which contradicts Lemma~\ref{lemma2}.
Therefore (\ref{14}) reduces to (\ref{31}).
\end{proof}

Let us now formulate the exact closure relations and prove our main result.

\section{Closure relation and proof of the main result}
\label{s9}

The exactness of (\ref{0.2}) is straightforward to see provided that outbreaks of
epidemics are always initiated with a single infected individual. We prove this first before considering the general case.

\subsection{Proof for single initial infected} \label{s4}

When infection is initiated on a tree graph at a single individual, infection must always
proceed in linear chains. Consequently there is no possibility of the state $I_kS_jI_i$
illustrated in Figure~\ref{figm1}
\begin{figure}
    \centerline{\includegraphics[width=0.25\textwidth]{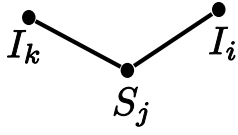}}
    \caption{Shown is a state which cannot arise on a tree graph where there is only one initially infectious node}
    \label{figm1}
\end{figure}
arising because an infection initiated at either $k$ or $i$ must pass through $j$ to get to the
other node. Furthermore, \beq \la S_jI_k\ra=\la S_iS_jI_k\ra+\la I_iS_jI_k\ra+\la R_iS_jI_k\ra,
\nonumber \eeq but since $\la I_iS_jI_k\ra=0$ and consequently $\la R_iS_jI_k\ra=0$, we have:
\beq \la S_jI_k\ra=\la S_iS_jI_k\ra \nonumber \eeq reducing (\ref{0.3}) to the following
closed system:
\begin{eqnarray}\nonumber
\dot{\la S_i\ra}&=&-\sum_{j=1}^P T_{ij}\la S_iI_j\ra, \\ \nonumber \dot{\la I_i\ra}&=&\sum_{j=1}^P
T_{ij}\la S_iI_j\ra -\gamma_i\la I_i\ra,\\ \nonumber
\dot{\la S_iI_j\ra}&=&\sum_{k=1, k\neq i}^P T_{jk}\la S_jI_k\ra-T_{ij}\la S_iI_j\ra-\gamma_j\la S_iI_j\ra, \\ \nonumber
\dot{\la S_iS_j\ra}&=&-\sum_{k=1, k\neq j}^PT_{ik}\la I_kS_i\ra-\sum_{k=1,k\neq i}^PT_{jk}\la
S_jI_k\ra. \label{3.13}
\end{eqnarray}
Similar arguments show that this can be written in the form of (\ref{0.2}).

More generally, this argument also applies to any tree graph where there is at most one network path by
which any susceptible individual in the network can become infectious from the initial
configuration of infected individuals.

Before discussing the general proof for any tree graph with multiple initially infected individuals, we consider
two very simple example networks which will serve to motivate and illustrate the method of proof.

\subsection{Proof for an open triple} \label{s5}

Here we consider the case for the open triple depicted in Figure~\ref{figm2}. The equations for the probabilities of the basic state space $M$ are given in (\ref{3.14}), (\ref{3.15}) and (\ref{3.16}). To form the relevant closure relations, we require the equations for the extended state space $\bar{M}$ formed by the equations for $M$ together with (\ref{3.141}),
\beq
\dot{\la S_1S_2\ra}=-\tau\la S_1S_2I_3\ra \;\;\;\; \textrm{    and    }\;\;\;\;\dot{\la S_2S_3\ra}=-\tau\la I_1S_2S_3\ra.    
\label{3.24}
\eeq
Our objective is to close the system at the level of pairs using the closure relation (\ref{kw}), eliminating the need for differential equations describing triples (\ref{3.16}), and show that the system remains exact. We note that the exactness of (\ref{0.2}) can be proved in this case along the lines of the previous argument by considering each possible initial condition separately; however, the approach discussed here will be more useful for understanding the general case.

We need to consider closures for the triples $\la I_1S_2I_3\ra$, $\la I_1S_2S_3\ra$ and $\la S_1S_2I_3\ra$. Let us consider the closure:
\beq \la I_1S_2I_3\ra\approx\frac{\la I_1S_2\ra\la S_2I_3\ra}{\la
S_2\ra}. \nonumber \eeq This is exact if $\alpha(t)=0$ where \beq \alpha(t)=\la S_2\ra\la
I_1S_2I_3\ra-\la I_1S_2\ra\la S_2I_3\ra \nonumber \eeq and $\la S_2\ra\neq 0$. Taking the derivative of $\alpha$ with
respect to time gives \beq \dot{\alpha}(t)=\dot{\la S_2\ra}\la I_1S_2I_3\ra+\la S_2\ra\dot{\la
I_1S_2I_3\ra}-\dot{\la I_1S_2\ra}\la S_2I_3\ra-\la I_1S_2\ra\dot{\la S_2I_3\ra}. \nonumber \eeq
Substituting the relevant derivatives in from (\ref{3.141})-(\ref{3.16}) and cancelling terms reduces this to
\beq \dot{\alpha}(t)=-2(\tau+\gamma )\alpha(t) , \nonumber\eeq so:
\beq
\alpha(t)=\alpha(0)e^{-2(\tau+\gamma)t}. \nonumber \eeq 
Now it is easily verified that provided
the system is initiated in a specific system state then $\alpha(0)=0$. Consequently $\alpha(t)=0$
for all $t\geq 0$ and the closure is exact.

By symmetry, it will suffice to consider one of the remaining two triples in (\ref{3.15}).
We wish to show that $\alpha(t)=0$ where \beq \alpha(t)=\la S_2\ra\la S_1S_2I_3\ra-\la
S_1S_2\ra\la S_2I_3\ra. \nonumber \eeq Here it is necessary to also use (\ref{3.24}) for pairs of type SS in the extended state space. This closure is not established immediately, but there is a two-step process to establishing that $\alpha(t)=0$ which the reader can verify by analogy with the example of the star graph in the next section.

\subsection{Proof for a star graph} \label{s6}

We now consider the case of the undirected star graph with $P=4$ shown in Figure~\ref{figm3},
\begin{figure}
    \centerline{\includegraphics[width=0.25\textwidth]{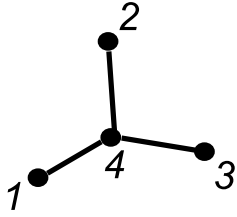}}
    \caption{Star graph with $P=4$ nodes}
    \label{figm3}
\end{figure}
where again we assume that the strength is the same across each network link and is denoted by
$\tau$ and the removal rate for each node is $\gamma$. Writing down the equations of the extended state space, there are two types of closure which need to be proved: one for the $S-S-I$ triples and one for the
$I-S-I$ triples (see (\ref{0.3})). The graph
has three triples $\left ( (1,4,3),(2,4,3),(1,4,2)\right )$, but it is sufficient to prove
exactness for one of them. Hence we want to prove the following two relations:
\begin{eqnarray}\nonumber
\la S_4\ra\la S_1I_3S_4\ra&=&\la S_1S_4\ra\la I_3S_4\ra, \\
\la S_4\ra\la I_1I_3S_4\ra&=&\la I_1S_4\ra\la I_3S_4\ra.
\label{4.11}
\end{eqnarray}
For brevity, we adopt the alternative notation:
\begin{eqnarray}\nonumber
\la\psi_4^S\ra\la\psi_{1,3,4}^{SIS}\ra-\la\psi_{1,4}^{SS}\ra\la\psi_{3,4}^{IS}\ra&=&0, \\ \nonumber
\la\psi_4^S\ra\la_{1,3,4}^{IIS}\ra-\la\psi_{1,4}^{IS}\ra\la\psi_{3,4}^{IS}\ra&=&0.
\end{eqnarray}
We introduce:
\beq
\alpha_1=\la\psi_4^S\ra\la\psi_{1,3,4}^{SIS}\ra-\la\psi_{1,4}^{SS}\ra\la\psi_{3,4}^{IS}\ra.
\nonumber
\eeq
By differentiating this, substituting in from the process equations and grouping terms, we obtain
\beq
\dot{\alpha}_1=-(\tau+\gamma)\alpha_1-\tau\alpha_2-\tau\alpha_3-\tau\alpha_4,
\label{4.14}
\eeq
where:
\begin{eqnarray}\nonumber
\alpha_2&=&\la\psi_{1,4}^{IS}\ra\la\psi_{1,3,4}^{SIS}\ra-\la\psi_{1,4}^{SS}\ra\la\psi_{1,3,4}^{IIS}\ra, \\ \nonumber
\alpha_3&=&\la\psi_{2,4}^{IS}\ra\la\psi_{1,3,4}^{SIS}\ra-\la\psi_{1,4}^{SS}\ra\la\psi_{2,3,4}^{IIS}\ra, \\\nonumber
\alpha_4&=&\la\psi_4^S\ra\la\psi_{1,2,3,4}^{SIIS}\ra-\la\psi_{1,2,4}^{SIS}\ra\la\psi_{3,4}^{IS}\ra.
\end{eqnarray}
Differentiating $\alpha_2$ we get
\beq
\dot{\alpha}_2=-2(\tau+\gamma )\alpha_2-\tau\alpha_5-\tau\alpha_6,
\label{4.16}
\eeq
where:
\begin{eqnarray}\nonumber
\alpha_5&=&\la\psi_{1,2,4}^{IIS}\ra\la\psi_{1,3,4}^{SIS}\ra-\la\psi_{1,2,4}^{SIS}\ra\la\psi_{1,3,4}^{IIS}\ra, \\ \nonumber
\alpha_6&=&\la\psi_{1,4}^{IS}\ra\la\psi_{1,2,3,4}^{SIIS}\ra-\la\psi_{1,4}^{SS}\ra\la\psi_{1,2,3,4}^{IIIS}\ra.
\end{eqnarray}
The derivatives of $\alpha_3$ and $\alpha_4$ can be obtained similarly.

Differentiating $\alpha_5$ we get
\beq
\dot{\alpha}_5=-3(\tau+\gamma )\alpha_5-\tau\alpha_7-\tau\alpha_8,
\label{4.18}
\eeq
where:
\begin{eqnarray}\nonumber
\alpha_7&=&\la\psi_{1,2,3,4}^{IIIS}\ra\la\psi_{1,3,4}^{SIS}\ra-\la\psi_{1,2,3,4}^{SIIS}\ra\la\psi_{1,3,4}^{IIS}\ra, \\ \nonumber
\alpha_8&=&\la\psi_{1,2,4}^{IIS}\ra\la\psi_{1,2,3,4}^{SIIS}\ra-\la\psi_{1,2,4}^{SIS}\ra\la\psi_{1,2,3,4}^{IIIS}\ra.
\end{eqnarray}
The derivative for $\alpha_6$ can also be obtained. Finally, differentiating $\alpha_7$ and $\alpha_8$ we obtain:
\beq
\dot{\alpha}_7=-4(\tau+\gamma)\alpha_7 \;\;\;\; \textrm{    and    } \;\;\;\; \dot{\alpha}_8=-4(\tau+\gamma)\alpha_8.
\label{4.21}
\eeq
To conclude the proof of the exactness of the closure, we first assume that the initial state is
not mixed; that is, one of the $3^4=81$ possible configurations has probability 1 at $t=0$. Then
it is easy to see that $\alpha_j(0)=0$ for all $j\in\{1,2,...,8\}$ (see Lemma~\ref{prop10} in
Section~\ref{sec3.5} for a proof in a more general context). Hence the differential equations for $\alpha_7$ and
$\alpha_8$ show that $\alpha_7(t)=0$ and $\alpha_8(t)=0$ for all $t\geq 0$. The differential equation
for $\alpha_5$ then implies that $\alpha_5=0$ $\forall t\geq 0$ (and similarly for $\alpha_6$). This
implies that $\alpha_2=0$ (and similarly $\alpha_3=\alpha_4=0$). The differential equation for
$\alpha_1$ shows that $\alpha_1=0$ which is what we wanted to show. The other triple closure in
(\ref{4.11}) can be proved similarly.
\begin{rem}
In fact, we have proved several closure relations $\alpha_j=0$ $\forall j\in \{1,2,...,8\}$.
\end{rem}

The closure relations each consist of two pairs which are visualised in Figure~\ref{fig0}. For
reference, we refer to these as the left pair and the right pair referring to their position in
this figure.
\begin{figure}
    \centerline{\includegraphics[width=0.75\textwidth]{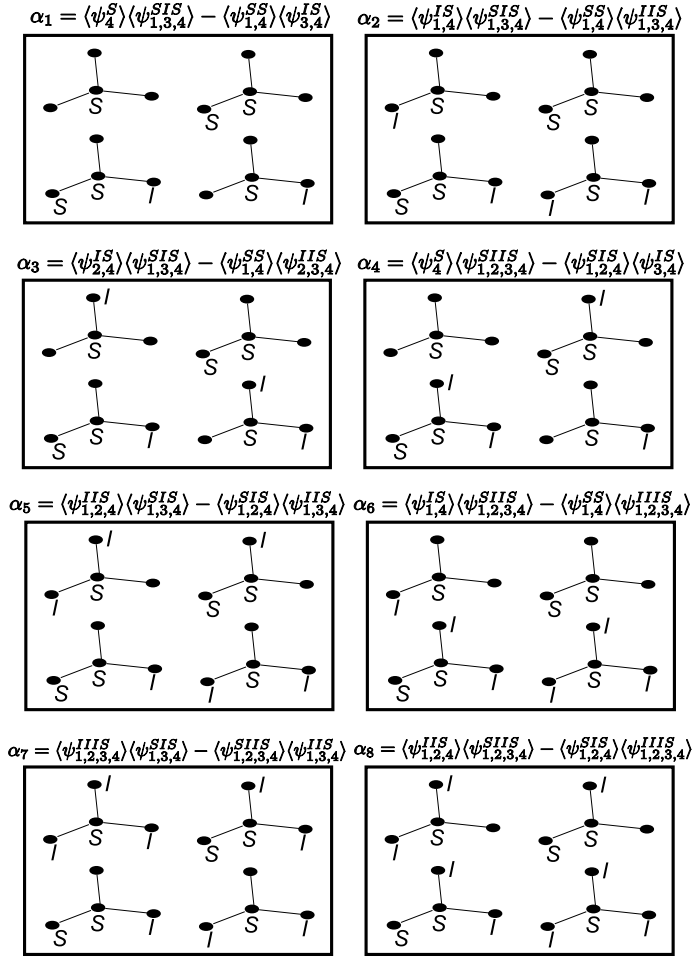}}
    \caption{Each box illustrates the relevant node states for the four parts of the closure relation in the equation above it. The node states on the left and the right correspond to the two terms in the closure relation. The node numbers correspond to the same positions as in Figure~\ref{figm3}}
\label{fig0}
\end{figure}
Looking at these closure relations, we can form two observations:
\begin{enumerate}
\item For a given node $i$, the number of times it appears as $S_i$ is the same in the left and the right pair, and similarly with the number of times it appears as $I_i$. For example, with $\alpha_5$, node 1 (the left node) has one $I$ and one $S$ for both pairs and node 4 (central node) has two $S$'s in both pairs. For $\alpha_6$, the number of $I$'s at nodes 1,2,3 and 4 is $(1,1,1,0)$ in both pairs and the number of $S$'s is given by $(1,0,0,2)$ in both pairs.
\item Any SI pairing on the left appears exactly the same number of times on the right. For example in $\alpha_7$, $I_1S_4$ appears once on the left and once on the right and $I_3S_4$ appears twice on the left and twice on the right. Observing that only $SS$ pairs and $IS$ pairs appear in the closure relations, a consequence is that $SS$ pairs also have this property.
\end{enumerate}
These observations will be of key importance for developing the general proof in the following two sections.

\subsection{General closure relations} \label{s9.1}

In general, to show that the closure relationship (\ref{kw}) is exact for the tree
graph, we need to show that $\alpha=0$ where 
\beq
 \alpha=\la B_j\ra\la A_iB_jC_k\ra-\la
A_iB_j\ra\la B_jC_k\ra,\nonumber
 \eeq 
 $B=S$ and $A,C\in \{S,I\}$.
Our proof of this is via induction using a sequence of closures
analogous to the proof in the case of the star graph in Section~\ref{s6}.

We shall consider many closure relations. In general we specify that they are composed of two pairs of motif states $(\psi_W^A,\psi_X^B)$ and $(\psi_Y^C,\psi_Z^D)$ and that the closure is exact if $\alpha=0$ where
\beq
\alpha=\la \psi_W^A\ra\la\psi_X^B\ra-\la\psi_Y^C\ra\la\psi_Z^D\ra.
\nonumber
\eeq

We formalise the observations we made about the closure relations for the star graph at the end of Section~\ref{s6} by defining what we term ``compatible pairs''.

\begin{defn}

\noindent For all $a\in \{S,I,R\} \textrm{ and } i\in \{1,2,...,P\}, \\ \psi_i^a\subset\psi_W^A\Leftrightarrow\exists j \textrm{  s.t.  } W_j=i\textrm{  and
} A_j=a$.

\noindent For all $a_1, a_2\in \{S,I,R\} \textrm{ and } i_1, i_2\in \{1,2,...,P\} : i_1\neq i_2, \\ \psi^{a_1a_2}_{i_1,i_2}\subset\psi_W^A\Leftrightarrow\exists j_1,j_2 \textrm{  s.t.  }
W_{j_1}=i_1,W_{j_2}=i_2\textrm{  and  } A_{j_1}=a_1,A_{j_2}=a_2$.
\end{defn}

\begin{rem}
In general, the notation $\psi_X^B\subset\psi_W^A$ denotes that the state of subsystem $\psi_X^B$ is implied by the state of subsystem $\psi_W^A$ because it is contained within it.
\end{rem}
\begin{defn}
\noindent Two pairs of motif states $(\psi_W^A,\psi_X^B)$ and $(\psi_Y^C,\psi_Z^D)$ are called compatible pairs if the following conditions are met:
\begin{itemize}
\item CP(i) $\left (\psi_i^a\subset\psi_W^A\textrm{ or }\psi_i^a\subset\psi_X^B\right )\Leftrightarrow\left (\psi_i^a\subset\psi_Y^C\textrm{ or }\psi_i^a\subset\psi_Z^D\right )$

\item CP(ii) $\left (\psi_i^a\subset\psi_W^A\textrm{ and }\psi_i^a\subset\psi_X^B\right )\Leftrightarrow\left (\psi_i^a\subset\psi_Y^C\textrm{ and }\psi_i^a\subset\psi_Z^D\right )$

\item CP(iii) $\left (\psi^{IS}_{i_1,i_2}\subset\psi_W^A\textrm{ or }\psi^{IS}_{i_1,i_2}\subset\psi_X^B\right )\Leftrightarrow\left (\psi^{IS}_{i_1,i_2}\subset\psi_Y^C\textrm{ or } \psi^{IS}_{i_1,i_2}\subset\psi_Z^D\right )$

\item CP(iv) $\left (\psi^{IS}_{i_1,i_2}\subset\psi_W^A\textrm{ and }\psi^{IS}_{i_1,i_2}\subset\psi_X^B\right )\Leftrightarrow\left (\psi^{IS}_{i_1,i_2}\subset\psi_Y^C\textrm{ and } \psi^{IS}_{i_1,i_2}\subset\psi_Z^D\right )$

\item CP(v) Same as CP(iii) and CP(iv) but with $SS$ pairs
\end{itemize}
where $a\in\{S,I,R\}$.
\end{defn}
\begin{defn}

Let $\psi_W^A$ be an $r$-state and $\psi_X^B$ be a $q$-state. Then the order of the pair $(\psi_W^A,\psi_X^B)$ is defined as $r+q$.
\end{defn}
\begin{prop}

If $(\psi_W^A,\psi_X^B)$ and $(\psi_Y^C,\psi_Z^D)$ are compatible pairs, then their order is equal.
\label{prop8}
\end{prop}

\begin{proof}

Follows from CP(i) and CP(ii).
\end{proof}
\begin{prop}

For a tree graph, applying the transformation $h_i$ to each of the four motif states in compatible pairs that contain node $i$ generates compatible pairs.
\label{prop9}
\end{prop}
\begin{proof}

The transformation satisfies CP(i) and CP(ii) because it replaces $\psi_i^a=\psi_i^I$ with $\psi_i^a=\psi_i^S$ which does not alter the form of the conditions. The transformation satisfies CP(iii) and CP(iv) because all $IS$ pairs where $i$ is the infected individual are removed by this transformation. New $IS$ pairs cannot be created by the transformation since this would require $II$ pairs which are prohibited for tree graphs by Lemma~\ref{lemma2}. CP(v) is satisfied because the transformation leaves existing $SS$ pairs unchanged and created $SS$ pairs result from existing $IS$ pairs so are balanced on each side.
\end{proof}

\subsection{Proof of the main result}
\label{sec3.5}
\begin{lem}

Let $(\psi_W^A,\psi_X^B)$ and $(\psi_Y^C,\psi_Z^D)$ be compatible pairs or order $R$ and $\psi_W^A,\psi_X^B,\psi_Y^C,\psi_Z^D\in\bar{M}$. Let
\beq
\alpha_0=\la\psi_W^A\ra\la\psi_X^B\ra-\la\psi_Y^C\ra\la\psi_Z^D\ra.
\nonumber
\eeq

Then:
\beq
\dot{\alpha_0}=\sum_{p=1}^mc_p\alpha_p+c_0\alpha_0,
\label{34}
\eeq
where each $\alpha_p$ can be expressed as
\beq
\alpha_p=\la \psi^{\bar{A}}_{\bar{W}}\ra\la\psi^{\bar{B}}_{\bar{X}}\ra-\la\psi^{\bar{C}}_{\bar{Y}}\ra\la\psi^{\bar{D}}_{\bar{Z}}\ra
\nonumber
\eeq
with $\psi^{\bar{A}}_{\bar{W}},\psi^{\bar{B}}_{\bar{X}}$ and $\psi^{\bar{C}}_{\bar{Y}},\psi^{\bar{D}}_{\bar{Z}}$ being compatible pairs of order $R+1$, and $c_0$, $c_p$ being constants and $m$ being an integer denoting the number of terms in the summation.
\label{lemma3}
\end{lem}
\begin{rem}
This is a general statement of the forms of (\ref{4.14})-(\ref{4.21}) in the star graph example.
\end{rem}
\begin{proof}

Take the derivative of $\alpha_0$:
\beq
\dot{\alpha_0}=\la\dot{\psi}^A_W\ra\la\psi_X^B\ra +  \la\psi_W^A\ra\la\dot{\psi}^B_X\ra   -\la\dot{\psi}_Y^C\ra\la\psi_Z^D\ra-\la\psi_Y^C\ra\la\dot{\psi}^D_Z\ra.
\label{36}
\eeq
We consider the terms associated with removal, transmission terms of order $R$ and transmission terms of order $R+1$ separately. Firstly, from (\ref{31}), this derivative contains the following terms associated with the removal process:
\begin{eqnarray}\nonumber
&&-\sum_{k_1}D_{k_1}^{AI}\gamm_{W_{k_1}}\la\psi_W^A\ra\la\psi_X^B\ra-\sum_{k_2}D_{k_2}^{BI}\gamm_{X_{k_2}}\la\psi_W^A\ra\la\psi_X^B\ra \\ \nonumber
&&+\sum_{k_3}D_{k_3}^{CI}\gamm_{Y_{k_3}}\la\psi_Y^C\ra\la\psi_Z^D\ra+\sum_{k_4}D_{k_4}^{DI}\gamm_{Z_{k_4}}\la\psi_Y^C\ra\la\psi_Z^D\ra \\ \nonumber
&=&-v\alpha_0,
\end{eqnarray}
where the sums over $k_1,k_2,k_3,k_4$ are over all nodes in the motifs $\psi_W,\psi_X,\psi_Y,\psi_Z$ respectively and where
\beq
v=\sum_{k_1}D_{k_1}^{AI}\gamm_{W_{k_1}}+\sum_{k_2}D_{k_2}^{BI}\gamm_{X_{k_2}}=\sum_{k_3}D_{k_3}^{CI}\gamm_{Y_{k_3}}+\sum_{k_4}D_{k_4}^{DI}\gamm_{Z_{k_4}}
\nonumber
\eeq
is easily seen to follow from CP(i) and CP(ii).

The right-hand side of (\ref{36}) also contains the following transmission terms with motifs of order $R$:
\begin{eqnarray}\nonumber
&&-\sum_{k_1}D_{k_1}^{AS}\sum_{l_1}T_{W_{k_1}W_{l_1}}D_{l_1}^{AI}\la\psi_W^A\ra\la\psi_X^B\ra -\sum_{k_2}D_{k_2}^{BS} \sum_{l_2}T_{X_{k_2}X_{l_2}}D_{l_2}^{BI}\la\psi_W^A\ra\la\psi_X^B\ra \\ \nonumber
&&+\sum_{k_3}D_{k_3}^{CS} \sum_{l_3}T_{Y_{k_3}Y_{l_3}}D_{l_3}^{CI}\la\psi_Y^C\ra\la\psi_Z^D\ra +\sum_{k_4}D_{k_4}^{DS}\sum_{l_4}T_{Z_{k_4}Z_{l_4}}D_{l_4}^{DI}\la\psi_Y^C\ra\la\psi_Z^D\ra \\ \nonumber
&=&-wa_0,
\end{eqnarray}
where
\begin{eqnarray}\nonumber
w&=&\sum_{k_1}D_{k_1}^{AS}\sum_{l_1}T_{W_{k_1}W_{l_1}}D_{l_1}^{AI}+\sum_{k_2}D_{k_2}^{BS}\sum_{l_2}T_{X_{k_2}X_{l_2}}D_{l_2}^{BI} \\ \nonumber
&=&\sum_{k_3}D_{k_3}^{CS}\sum_{l_3}T_{Y_{k_3}Y_{l_3}}D_{l_3}^{CI}+\sum_{k_4}D_{k_4}^{DS}\sum_{l_4}T_{Z_{k_4}Z_{l_4}}D_{l_4}^{DI}
\end{eqnarray}
and where the sums over $k_1,k_2,k_3,k_4,l_1,l_2,l_3,l_4$ are over all nodes in each of the relevant motifs. This follows from CP(iii) and CP(iv). Hence the removal terms and transmission terms of order $R$ contribute $c_0a_0$ to the derivative of $a_0$ where $c_0=-v-w$.

For transmission terms with motifs of order $R+1$, consider the term $\la\dot{\psi}^A_W\ra\la\psi_X^B\ra$  in (\ref{36}). This gives rise to the following terms in the derivative of $\alpha_0$:
\beq
\sum_{k_1}(-1)^{D_{k_1}^{AS}}\sum_{n=1,n\notin W}^PT_{W_{k_1}n}\la g_{W_{k_1}}^n(\psi_W^A)\ra\la\psi_X^B\ra.
\nonumber
\eeq
To prove the lemma, it is sufficient to show that each term in this sum can be paired uniquely with a term in $\la\dot{\psi}_Y^C\ra\la\psi_Z^D\ra$ or in $\la\psi_Y^C\ra\la\dot{\psi}^D_Z\ra$, such that the difference of these terms forms $\alpha_p$. By symmetry this is a one-to-one pairing establishing that each term of order $R+1$ on the right-hand side of (\ref{36}) is accounted for exactly once in the sum of $\alpha_p$.

Let us take an element from the sum by choosing a node $W_w$, $w\in\{1,2,...,r\}$ and an outside neighbour node $n\in N_{W_w}$. This neighbour can either be in $\psi_X$ or outside. 

\noindent \it Case 1: $A_w=S$\rm

Consider first the case where $A_w$ is a susceptible node, resulting in the following term in the sum:
\beq
-T_{{W_w}n}\la g_{W_w}^n(\psi^A_W)\ra\la\psi_X^B\ra,
\nonumber
\eeq
where we can identify $c_p=-T_{{W_w}n}$.

We have $A_w=S$, $n\in N_{W_w}$ and $n\notin W$. Let us now identify a term in $\la\dot{\psi}_Y^C\ra\la\psi_Z^D\ra+\la\psi_Y^C\ra\la\dot{\psi}^D_Z\ra$ to form compatible pairs. According to CP(i) and CP(ii) we can assume without loss of generality that $W_w\in Y$, i.e. $\exists y : Y_y=W_w$ and $C_y=S$ (see Figure~\ref{fig1}).
\begin{figure}
\centerline{\includegraphics[width=.9\textwidth]{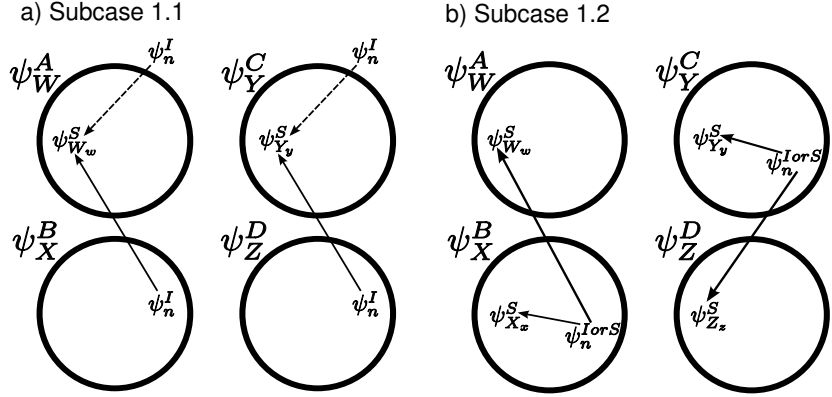}}
\caption{Each circle refers to one of the motif states $\psi_W^A,\psi_X^B,\psi_Y^C,\psi_Z^D$ specified to the top left. The position of the relevant node states with respect to the motif states are then illustrated. a) Subcase 1.1 ( $n\notin Y$). b) Subcase 1.2 ($n\in Y$)}
\label{fig1}
\end{figure}
There are two subcases:

\noindent\it Subcase 1.1: $n\notin Y$\rm

When $n\notin Y$, either $n\in Z$ or $n\notin Z$ and these are shown by solid and dashed lines respectively in Figure~\ref{fig1}a. By CP(i), $n\in Z\Leftrightarrow n\in X$ since $n\notin W$ so the solid lines match on the left and right pairs as do the dashed lines.
The corresponding term must therefore be $-T_{Y_yn}\la g_{Y_y}^n(\psi_Y^C)\ra\la\psi_Z^D\ra=c_p\la g_{Y_y}^n(\psi_Y^C)\ra\la\psi_Z^D\ra$ irrespective of whether $n\in Z$ or not. Hence:
\beq
\alpha_p=\la g_{W_w}^n(\psi_W^A)\ra\la\psi^B_X\ra-\la g_{Y_y}^n(\psi^C_Y)\ra\la\psi^D_Z\ra,
\nonumber
\eeq
where $(g_{W_w}^n(\psi^A_W),\psi^B_X)$ and $(g_{Y_y}^n(\psi^C_Y),\psi^D_Z)$ are easily seen to satisfy the definition of compatible pairs since the extra node is $n$ which is $I$ in both pairs.

\noindent\it Subcase 1.2: $n\in Y$\rm

If $n\in Y$, then the edge $n\rightarrow Y_y$ is an $SS$ or $IS$ edge in $C$. By CP(iii), CP(iv) and CP(v), it is also the same edge in $B$ because $n\notin W$. Hence $\exists x : X_x=W_w \textrm{ and } B_x=S$. By CP(ii), $W_w\in Z$ is also true where $\exists z : Z_z=W_w \textrm{ and } D_z=S$. This is illustrated in Figure~\ref{fig1}b.
We therefore have the corresponding term $-T_{Z_zn}\la\psi^C_Y\ra\la g_{Z_z}^n(\psi^D_Z)\ra=c_p\la\psi^C_Y\ra\la g_{Z_z}^n(\psi^D_Z)\ra$ and:
\beq
\alpha_p=\la g_{W_w}^n(\psi^A_W)\ra\la \psi^B_X\ra-\la \psi^C_Y\ra\la g_{Z_z}^n(\psi^D_Z)\ra,
\nonumber
\eeq
where again, the relevant pairs are seen to satisfy compatibility.

\noindent \it Case 2: $A_w=I$\rm

So far we have proved the existence of $\alpha_p$ when $A_w=S$, $n\in N_{W_w}$, $n\notin W$. Now we have to show $\alpha_p$ can be defined when $A_w=I$, $n\in N_{W_w}$, $n\notin W$. In this case, the motif generating rule firstly changes $A_w=I$ to $A_w=S$ and then applies the same generating rule as if $A_w=S$ initially. From Proposition~\ref{prop9}, applying the transformation to compatible pairs of order $R$ produces compatible pairs of order $R$ in the case of the tree graph. After this transformation, the argument runs identically to case 1. This completes the proof of Lemma~\ref{lemma3}.
\end{proof}
\begin{lem}
Assume that the initial condition is not mixed, i.e. $\exists A\in\{I,S\}^P$ such that $\la\psi_{1,2,...,P}^A\ra=1$. If $(\psi_W^A,\psi_X^B)$ and $(\psi_Y^C,\psi_Z^D)$ are compatible pairs, then for any graph, $\la\psi_W^A\ra\la\psi_X^B\ra-\la\psi_Y^C\ra\la\psi_Z^D\ra=0$ at $t=0$.
\label{prop10}
\end{lem}
\begin{proof}
Assume that $\la\psi_W^A\ra\la\psi_X^B\ra=1$. Then by CP(i) and CP(ii), for all $\psi_i^a\subset\psi_W^A$, we must have $\psi_i^a\subset\psi_Y^C$ and/or $\psi_i^a\subset\psi_Z^D$. This is also true for all $\psi_i^a\subset\psi_X^B$ and, by the symmetry between the compatible pairs, it follows that $\la\psi_Y^C\ra\la\psi_Z^D\ra=1$. Similarly, it follows that $\la\psi_W^A\ra\la\psi_X^B\ra=0$ implies that $\la\psi_Y^C\ra\la\psi_Z^D\ra=0$.
\end{proof}

\begin{thm}

\noindent Let us assume the following:
\begin{itemize}
\item The graph is a tree.
\item The initial condition is not mixed.
\item $(\psi_W^A,\psi_X^B)$ and $(\psi_Y^C,\psi_Z^D)$ are compatible pairs.
\end{itemize}
Then $\la\psi_W^A\ra\la\psi_X^B\ra-\la\psi_Y^C\ra\la\psi_Z^D\ra=0$ for all time $t\geq0$.
\label{theorem2}
\end{thm}

\begin{proof} We prove the theorem by induction according to the order of the closure. This is analogous to the proof for the star graph in Section~\ref{s6}.

\noindent Step 1

If the closure is of order $2P$, then it is exact. More precisely, if $\psi_W^A$, $\psi_X^B$, $\psi_Y^C$ and $\psi_Z^D$ are $P$-states, then (\ref{34}) does not contain the summation terms and becomes:
\beq
\dot{\alpha_0}=c_0\alpha_0.
\nonumber
\eeq
Since we start from an initial condition that is not mixed, we have (by Lemma~\ref{prop10}) $\alpha_0(0)=0\Rightarrow \alpha_0(t)=0$ $\forall t\geq 0$.

\noindent Step 2

Assume that the theorem is proved for compatible pairs of order $R+1$. We prove that it is true for compatible pairs of order $R$. Applying Lemma~\ref{lemma3}, we have:
\beq
\dot{\alpha_0}=\sum_{p=1}^mc_p\alpha_p+c_0\alpha_0.
\nonumber
\eeq

According to the induction condition, $\alpha_p=0$ $\forall p$ because these are compatible pairs of order $R+1$. Therefore $\dot{\alpha_0}=c_0\alpha_0$. From Lemma~\ref{prop10}, $\alpha_0(0)=0$ so $\alpha_0(t)=0$ $\forall t\geq 0$. Since we have proved the result for compatible pairs of order $2P$ then we have completed the proof of the theorem.
\end{proof}

The lowest-order compatible pairs are of order four. The closure relation corresponding to these pairs is formulated in the following important corollary.
\begin{cor}

Under the assumptions on the graph and the initial conditions in Theorem~\ref{theorem2}, we have the special cases:
\beq
\la\psi_j^S\ra\la\psi_{i,j,\,k}^{SSI}\ra=\la\psi_{i,j}^{SS}\ra\la_{j,\,k}^{SI}\ra
\nonumber\eeq
for all $i\in\{1,2,...,P\}$ and for all $j\in N_i$, $k\in N_j$: $i\neq k$;
\beq
\la\psi_i^S\ra\la\psi_{k\!,\,i, j}^{ISI}\ra=\la\psi_{k\!,\,i}^{IS}\ra\la\psi_{i,j}^{SI}\ra
\nonumber\eeq
for all $i\in\{1,2,...,P\}$ and for all $j,k\in N_i$: $j\neq k$.

\end{cor}

This corollary is Theorem~\ref{theoremMain} expressed in a different notation.
\begin{rem}
From CP(i) and CP(ii), it is clear that Lemma~\ref{prop10} can be extended to the mixed initial condition where the probabilities of the initial states of each individual in the system are statistically independent, leading to $\la\psi_W^A\ra\la\psi_X^B\ra-\la\psi_Y^C\ra\la\psi_Z^D\ra=0$ at $t=0$. However, for general mixed initial conditions where correlations between individuals can occur, Lemma~\ref{prop10} does not hold and the pair-based model is not exact.
\end{rem}
\section{Application to some graphs which are not trees}
\label{s4.0}
To complete this work, we make a final observation which shows that the pair-based model can sometimes provide an exact representation of infectious dynamics on graphs which are not strictly trees. We first make two definitions which can be understood with reference to the examples in Figure~\ref{example_figures}. \begin{figure}
    \centerline{\includegraphics[width=0.50\textwidth]{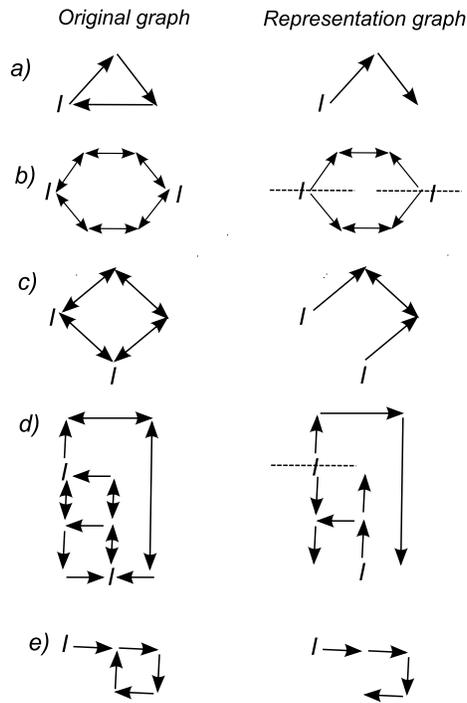}}
    \caption{The graphs on the left are the initial transmission networks where the initially infected nodes are indicated by the symbol $I$. The graphs on the right are the reduced representation graphs where the cuts for independent segments which occur for cases b and d are indicated with dashed lines. The tree structure of the graphs on the right shows that applying the pair-based model to these graphs generates an exact representation of the infection dynamics on the original system}
    \label{example_figures}
\end{figure} \begin{defn}
A reduced representation is a graph which is constructed from the initial transmission network and the given initial conditions by removing transmission routes which cannot carry infection dynamics.
\end{defn}
\begin{defn}
An independent segment is a region of a graph that is only connected to other regions via nodes in the segment which are initially infectious.
\end{defn}
\begin{thm}
Given SIR dynamics on a transmission network with infection and removal governed by Poisson processes and given an unmixed initial state of the system, if every independent segment of the reduced representation is a tree, then applying (\ref{0.2}) to this representation exactly generates the expected infection dynamics on the original transmission network.
\label{reduction_theorem}
\end{thm}
\begin{proof}
By definition, the infection dynamics of the system remain unchanged after the removal of edges which cannot support infection dynamics. Additionally, the infection dynamics of any independent segment are independent of the dynamics on the rest of the graph because there is no process that allows influence across the initially infectious nodes. If the resulting representation graph is a set of trees, then since (\ref{0.2}) is an exact representation of the dynamics on each independent segment, solving (\ref{0.2}) on the reduced representation graph is equivalent to the infection dynamics on the original transmission network.
\end{proof}
Figure~\ref{example_figures} shows some graphs and the associated representation graphs where the dashed lines indicate the boundaries that separate independent segments. For each of these examples, the solution of (\ref{0.2}) on the representation graph exactly reproduces the expected infection dynamics of the original system. 

This suggests that the accuracy of the pair-based model could be increased by first generating the representation graph for the particular network and initial conditions prior to numerically solving the pair-based model.
\section{Discussion}
\label{s10}
We considered the pair-based variant of the subsystem approach to constructing epidemic models on networks (Sharkey 2008, 2011). We proved that for SIR dynamics on fixed tree graphs with exponentially distributed transmission and removal processes, the pair-based model provides an exact determination of the infection probability time course for each individual in the network. We also showed that the dynamics of some networks with cycles can be represented exactly by the pair-based model under specific initial conditions.

This represents the first provably exact deterministic model of epidemic dynamics on finite heterogeneous systems which has been numerically evaluated. Here we use the qualifying term ``heterogeneous'' to exclude systems with significant symmetry which may be employed to obtain exact representations in very specialised circumstances (Keeling and Ross 2008; Simon et al. 2011). In principle, the message-passing approach of Karrer and Newman (2010) will also yield an exact description of finite heterogeneous systems in a way that is numerically feasible, but to our knowledge this has not yet been implemented in this context. Interestingly, the message-passing method also applies more generally beyond the usual assumptions of Markovian dynamics to arbitrary distributions for transmission and removal processes, although there may be implementation issues for more general distributions.

We note that effective degree models can generate very good agreement with stochastic simulation (Ball and Neal 2008; Lindquist et al. 2011) as do the PGF or edge-based compartmental modelling methods (Miller et al. 2012; Miller and Volz 2012; Volz 2008), although exact correspondence has not been proven here. For some idealised networks, including fully connected networks and some configuration networks (Volz 2008), convergence to the expected value can be shown in the infinite population limit (Ball and Neal 2008; Decreusefond et al. 2012; Karrer and Newman 2010). However, these models have a large measure of homogeneity, and convergence only occurs for infinite populations.

It is intuitively understood that clustering is at the root of problems with models based around closures at the level of pairs (Keeling and Eames 2005). Previous analysis (Sharkey 2011) attributed the failure to anomalous terms which emerge in subsystem equations when differentiating closure approximations based around the statistical independence of individuals. Here, repeating similar analysis for a closure at the order of pairs in the context of tree graphs, these anomalies do not arise and we are able to prove that the closure is exact via induction.

In principle, models based around subsystems at the order of three nodes or higher could be constructed. The next higher-order model would require obtaining a closure which is able to preserve correlations between triples, and similarly for higher orders. This leads to an interesting theoretical question for future analysis: does the hierarchy of exact order-by-order models suggested in Sharkey (2011) exist, and if so, what form should the closure approximations take at each level? We conjecture that exact closures of a similar nature to those considered here are possible for networks with more structure, given that the order at which the closure is performed is guided by the network structure; future work will focus on this question.

\section*{Appendix}
The proof of Theorem~\ref{theorem1} is analogous to the proof of the single and pair equations by Sharkey (2011) in~\cite{Sharkey11}. In what follows, summations over Greek indices $\alpha$,$\beta$ are assumed to be over all $3^P$ possible system states. First we make some definitions.

\begin{defn}
For a system $\Gamma$ in state $\alpha$ and a single node $i$ of $\Gamma$ in state $a$ we define:
\beq
D_i^{\alpha a}=\left\{
\begin{array}{ll}
1 & \textrm{if $\psi_i^a\subset\Gamma^\alpha $,} \\ 
0 & \textrm{otherwise,} 
\end{array}\right.\nonumber
\eeq
denoting whether or not the specified single node state matches the system state. Note that this is just Definition~\ref{D_def} applied to the full system.
\end{defn}

\begin{defn}

\beq
\zeta_j^{\alpha\beta}=\left\{
\begin{array}{ll}
1 & \textrm{if the states of all individuals in $\Gamma$ are the same} \\
& \textrm{for $\Gamma^\alpha$ and for $\Gamma^\beta$ except for $\psi_j$ which may change,}  \\
0 & \textrm{otherwise.}
\end{array}\right.
\nonumber
\eeq
\end{defn}
\begin{prop}
For all $\alpha, i$:
\beq
\sum_a D_i^{\alpha a}=1,
\nonumber
\eeq
where the summation is over all possible states available to node $i$.
\label{prop3}
\end{prop}
\begin{proof}
Statement that for a given system state $\Gamma^\alpha$, or subsystem state $\psi_W^A$, each node must be in a unique state.
\end{proof}

\begin{prop}
For all $\beta,i,a$:
\beq
\sum_\alpha D_i^{\alpha a}\zeta_i^{\alpha\beta}=1.
\nonumber
\eeq
\label{prop4}
\end{prop}
\begin{proof}
Statement that there is only one system state which is identical to $\Gamma^\beta$ except that node $i$ is in state $\psi_i^a$.
\end{proof}
\begin{prop}
For any subsystem $\psi_W^A$ and $\forall k\in\{1,2,...,r\}$:
\beq
D_{W_k}^{\alpha A_k}D_{W_k}^{\alpha a}=D_k^{Aa}D_{W_k}^{\alpha a}
\nonumber
\eeq
for all $\alpha,a$.
\label{prop5}
\end{prop}
\begin{proof}

Proposition is true when $D_{W_k}^{\alpha a}=0$. When $D_{W_k}^{\alpha a}=1$ we have:

$D_{W_k}^{\alpha A_k}=1\Leftrightarrow a=A_k\Leftrightarrow D_k^{Aa}=1$,

$D_{W_k}^{\alpha A_k}=0\Leftrightarrow a\neq A_k\Leftrightarrow D_k^{Aa}=0$.

\end{proof}
\begin{prop}
For any subsystem $\psi_W^A$ and $\forall k\in\{1,2,...,r\}$:
\beq
\sum_\alpha D_{W_k}^{\alpha a}\zeta_{W_k}^{\alpha\beta}\prod_{j=1,j\neq k}^rD_{W_j}^{\alpha A_j}D_{W_j}^{\beta A_j}=\prod_{j=1, j\neq k}^r D_{W_j}^{\beta A_j}
\nonumber
\eeq
for all $\beta,a$.
\label{prop6}
\end{prop}
\begin{proof}

Proposition is true when:
\beq
\prod_{j=1, j\neq k}^r D_{W_j}^{\beta A_j}=0.
\nonumber
\eeq
From Proposition~\ref{prop4} there must be a single state $\Gamma^\alpha$ for which $D_{W_k}^{\alpha a}\zeta_{W_k}^{\alpha\beta}=1$, otherwise it is zero. When
\beq
\prod_{j=1, j\neq k}^r D_{W_j}^{\beta A_j}=1,
\nonumber
\eeq
 we must also have (for the state when $D_{W_k}^{\alpha\beta}\zeta_{W_k}^{\alpha a}=1$):
\beq
\prod_{j=1, j\neq k}^r D_{W_j}^{\alpha A_j}=1,
\nonumber
\eeq
because only site $\psi_{W_k}$ can change state during this transition, establishing the proposition.
\end{proof}

We can now use these propositions to prove Theorem~\ref{theorem1}:

\begin{proof}

We have that:
\begin{eqnarray}
\la\psi_W^A\ra&=&\sum_\alpha\la\Gamma^\alpha\ra\prod_{i=1}^rD_{W_i}^{\alpha A_i}.
\nonumber
\end{eqnarray}
Taking the derivative of this with respect to time and substituting in the system master equation (\ref{eq3.1}) gives
\begin{eqnarray}\nonumber
\dot{\la\psi_W^A\ra}&=&\sum_\alpha\dot{\la\Gamma^\alpha\ra}\prod_{i=1}^rD_{W_i}^{\alpha A_i} \\
&=&\sum_{\alpha\beta}\sigma^{\alpha\beta}\la\Gamma^{\beta}\ra\prod_{i=1}^rD_{W_i}^{\alpha A_i}-\sum_{\alpha\beta}\sigma^{\beta\alpha}\la\Gamma^\alpha\ra\prod_{i=1}^rD_{W_i}^{\alpha A_i}.
\label{16}
\end{eqnarray}
From Proposition~\ref{prop3}:
\begin{eqnarray}\nonumber
1&=&\left [\sum_{a_1}D_{W_1}^{\alpha a_1}\right ]...\left [\sum_{a_r}D_{W_r}^{\alpha a_r}\right ]\left [\sum_{b_1}D_{W_1}^{\beta b_1}\right ]...\left [\sum_{b_r}D_{W_r}^{\beta b_r}\right ] \\
&=&\sum_{k=1}^r\sum_{a_kb_k}\prod_{j=1}^rD_{W_j}^{\alpha a_j}D_{W_j}^{\beta b_j}.
\nonumber
\end{eqnarray}
Multiplying the right of (\ref{16}) by this gives:
\begin{eqnarray}\nonumber
\dot{\la\psi_W^A\ra}&=&\sum_{\alpha\beta}\sigma^{\alpha\beta}\la\Gamma^{\beta}\ra\prod_{i=1}^rD_{W_i}^{\alpha A_i} \sum_{k=1}^r\sum_{a_kb_k}\prod_{j=1}^rD_{W_j}^{\alpha a_j}D_{W_j}^{\beta b_j}\\
&&-\sum_{\alpha\beta}\sigma^{\beta\alpha}\la\Gamma^\alpha\ra\prod_{i=1}^rD_{W_i}^{\alpha A_i}\sum_{k=1}^r\sum_{a_kb_k}\prod_{j=1}^rD_{W_j}^{\alpha a_j}D_{W_j}^{\beta b_j}.
\nonumber
\end{eqnarray}
This can be simplified using the fact that $\sigma^{\alpha\beta}=0$ whenever the state of the subsystem $\psi_W$ differs by more than a single individual $\psi_{W_k}$, $k\in \{1...r\}$ between states $\Gamma^\alpha$ and $\Gamma^\beta$ which means that $a_j=b_j=A_j$ for $j\neq k$:
\begin{eqnarray}\nonumber
\dot{\la\psi_W^A\ra}&=&\sum_{\alpha\beta}\sigma^{\alpha\beta}\la\Gamma^{\beta}\ra\prod_{i=1}^rD_{W_i}^{\alpha A_i} \sum_{k=1}^r\prod_{j=1, j\neq k}^rD_{W_j}^{\alpha A_j}D_{W_j}^{\beta A_j}\sum_{a_kb_k}D_{W_k}^{\alpha a_k}D_{W_k}^{\beta b_k} \\ \nonumber
&&-\sum_{\alpha\beta}\sigma^{\beta\alpha}\la\Gamma^\alpha\ra\prod_{i=1}^rD_{W_i}^{\alpha A_i}\sum_{k=1}^r\prod_{j=1, j\neq k}^rD_{W_j}^{\alpha A_j}D_{W_j}^{\beta A_j}\sum_{a_kb_k}D_{W_k}^{\alpha a_k}D_{W_k}^{\beta b_k} \\ \nonumber
&=&\sum_{\alpha\beta}\sigma^{\alpha\beta}\la\Gamma^{\beta}\ra \sum_{k=1}^rD_{W_k}^{\alpha A_k}\prod_{j=1, j\neq k}^rD_{W_j}^{\alpha A_j}D_{W_j}^{\beta A_j}\sum_{a_kb_k}D_{W_k}^{\alpha a_k}D_{W_k}^{\beta b_k} \\ \nonumber
&&-\sum_{\alpha\beta}\sigma^{\beta\alpha}\la\Gamma^\alpha\ra\sum_{k=1}^rD_{W_k}^{\alpha A_k}
\prod_{j=1, j\neq k}^rD_{W_j}^{\alpha A_j}D_{W_j}^{\beta A_j}\sum_{a_kb_k}D_{W_k}^{\alpha a_k}D_{W_k}^{\beta b_k},
\end{eqnarray}
where the last equality follows from $D_{W_j}^{\alpha A_j}D_{W_j}^{\alpha A_j}=D_{W_j}^{\alpha A_j}$.

For SIR dynamics, we can do the summations over $a_k$ and $b_k$:
\begin{eqnarray}\nonumber
\dot{\la\psi_W^A\ra}&=&\sum_{\alpha\beta}\sigma^{\alpha\beta}\la\Gamma^\beta\ra\sum_{k=1}^rD_{W_k}^{\alpha A_k}\prod_{j=1, j\neq k}^rD_{W_j}^{\alpha A_j}D_{W_j}^{\beta A_j}D_{W_k}^{\alpha I}D_{W_k}^{\beta S} \\ \nonumber
&&\sum_{\alpha\beta}\sigma^{\alpha\beta}\la\Gamma^\beta\ra\sum_{k=1}^rD_{W_k}^{\alpha A_k}\prod_{j=1, j\neq k}^rD_{W_j}^{\alpha A_j}D_{W_j}^{\beta A_j}D_{W_k}^{\alpha R}D_{W_k}^{\beta I} \\ \nonumber
&&-\sum_{\alpha\beta}\sigma^{\beta\alpha}\la\Gamma^{\alpha}\ra\sum_{k=1}^rD_{W_k}^{\alpha A_k}\prod_{j=1,j\neq k}^rD_{W_j}^{\alpha A_j}D_{W_j}^{\beta A_j}D_{W_k}^{\alpha S}D_{W_k}^{\beta I} \\
&&-\sum_{\alpha\beta}\sigma^{\beta\alpha}\la\Gamma^{\alpha}\ra\sum_{k=1}^rD_{W_k}^{\alpha A_k}\prod_{j=1,j\neq k}^rD_{W_j}^{\alpha A_j}D_{W_j}^{\beta A_j}D_{W_k}^{\alpha I}D_{W_k}^{\beta R}.
\label{20}
\end{eqnarray}

Now we introduce the relevant terms in the transition matrix at the level of the system:
\begin{eqnarray}\nonumber
\sigma^{\alpha\beta}D_{W_k}^{\alpha I}D_{W_k}^{\beta S}\prod_{j=1,j\neq k}^rD_{W_j}^{\alpha A_j}D_{W_j}^{\beta A_j}&=&\sum_{n=1}^PT_{W_kn}D_{W_k}^{\alpha I}D_{W_k}^{\beta S}D_n^{\beta I}\zeta_{W_k}^{\alpha\beta}\prod_{j=1,j\neq k}^rD_{W_j}^{\alpha A_j}D_{W_j}^{\beta A_j},  \\ \nonumber
\sigma^{\alpha\beta}D_{W_k}^{\alpha R}D_{W_k}^{\beta I}\prod_{j=1,j\neq k}^rD_{W_j}^{\alpha A_j}D_{W_j}^{\beta A_j}&=&\gamm_{W_k}D_{W_k}^{\alpha R}D_{W_k}^{\beta I}\zeta_{W_k}^{\alpha\beta}\prod_{j=1,j\neq k}^rD_{W_j}^{\alpha A_j}D_{W_j}^{\beta A_j}, \\ \nonumber
\sigma^{\beta\alpha}D_{W_k}^{\alpha S}D_{W_k}^{\beta I}\prod_{j=1,j\neq k}^rD_{W_j}^{\alpha A_j}D_{W_j}^{\beta A_j}&=&\sum_{n=1}^PT_{W_kn}D_{W_k}^{\alpha S}D_{W_k}^{\beta I}D_n^{\alpha I}\zeta_{W_k}^{\beta\alpha}\prod_{j=1,j\neq k}^rD_{W_j}^{\alpha A_j}D_{W_j}^{\beta A_j}, \\ \nonumber
\sigma^{\beta\alpha}D_{W_k}^{\alpha I}D_{W_k}^{\beta R}\prod_{j=1,j\neq k}^rD_{W_j}^{\alpha A_j}D_{W_j}^{\beta A_j}&=&\gamm_{W_k}D_{W_k}^{\alpha I}D_{W_k}^{\beta R}\zeta_{W_k}^{\beta\alpha}\prod_{j=1,j\neq k}^rD_{W_j}^{\alpha A_j}D_{W_j}^{\beta A_j},
\end{eqnarray}
where these equations are designed so that they are satisfied for any combination of $\alpha$, $\beta$, $k$. Substituting these into (\ref{20}) gives:
\begin{eqnarray}\nonumber
\dot{\la\psi_W^A\ra}&=&\sum_{\alpha\beta}\la\Gamma^\beta\ra\sum_{k=1}^rD_{W_k}^{\alpha A_k}\sum_{n=1}^PT_{W_kn}D_{W_k}^{\alpha I}D_{W_k}^{\beta S}D_n^{\beta I}\zeta_{W_k}^{\alpha\beta}\prod_{j=1,j\neq k}^rD_{W_j}^{\alpha A_j}D_{W_j}^{\beta A_j}  \\ \nonumber
&&+\sum_{\alpha\beta}\la\Gamma^\beta\ra\sum_{k=1}^rD_{W_k}^{\alpha A_k}\gamm_{W_k}D_{W_k}^{\alpha R}D_{W_k}^{\beta I}\zeta_{W_k}^{\alpha\beta}\prod_{j=1,j\neq k}^rD_{W_j}^{\alpha A_j}D_{W_j}^{\beta A_j} \\ \nonumber
&&-\sum_{\alpha\beta}\la\Gamma^{\alpha}\ra\sum_{k=1}^rD_{W_k}^{\alpha A_k}\sum_{n=1}^PT_{W_kn}D_{W_k}^{\alpha S}D_{W_k}^{\beta I}D_n^{\alpha I}\zeta_{W_k}^{\beta\alpha}\prod_{j=1,j\neq k}^rD_{W_j}^{\alpha A_j}D_{W_j}^{\beta A_j}\\ \nonumber
&&-\sum_{\alpha\beta}\la\Gamma^{\alpha}\ra\sum_{k=1}^rD_{W_k}^{\alpha A_k}\gamm_{W_k}D_{W_k}^{\alpha I}D_{W_k}^{\beta R}\zeta_{W_k}^{\beta\alpha}\prod_{j=1,j\neq k}^rD_{W_j}^{\alpha A_j}D_{W_j}^{\beta A_j}.
\end{eqnarray}
We can rearrange the summation order:
\begin{eqnarray}\nonumber
\dot{\la\psi_W^A\ra}&=&\sum_{k=1}^r\sum_{n=1}^PT_{W_kn}\sum_{\beta}\la\Gamma^\beta\ra D_{W_k}^{\beta S}D_n^{\beta I}\sum_\alpha D_{W_k}^{\alpha A_k}D_{W_k}^{\alpha I}\zeta_{W_k}^{\alpha\beta}\prod_{j=1, j\neq k}^rD_{W_j}^{\alpha A_j}D_{W_j}^{\beta A_j} \\ \nonumber
&&+\sum_{k=1}^r\gamm_{W_k}\sum_\beta \la\Gamma^{\beta}\ra D_{W_k}^{\beta I}\sum_\alpha D_{W_k}^{\alpha A_k}D_{W_k}^{\alpha R}\zeta_{W_k}^{\alpha\beta}\prod_{j=1,j\neq k}^rD_{W_j}^{\alpha A_j}D_{W_j}^{\beta A_j} \\ \nonumber
&&-\sum_{k=1}^r\sum_{n=1}^PT_{W_kn}\sum_\alpha\la\Gamma^{\alpha}\ra D_{W_k}^{\alpha A_k}D_{W_k}^{\alpha S}D_n^{\alpha I}\sum_\beta  D_{W_k}^{\beta I}\zeta_{W_k}^{\beta\alpha} \prod_{j=1,j\neq k}^rD_{W_j}^{\alpha A_j}D_{W_j}^{\beta A_j}\\ \nonumber
&&-\sum_{k=1}^r\gamm_{W_k}\sum_\alpha \la\Gamma^{\alpha}\ra D_{W_k}^{\alpha A_k}D_{W_k}^{\alpha I}\sum_\beta D_{W_k}^{\beta R}\zeta_{W_k}^{\beta\alpha}\prod_{j=1,j\neq k}^rD_{W_j}^{\alpha A_j}D_{W_j}^{\beta A_j},
\end{eqnarray}
and apply Proposition~\ref{prop5}:
\begin{eqnarray}\nonumber
\dot{\la\psi_W^A\ra}&=&\sum_{k=1}^rD_k^{AI}\sum_{n=1}^PT_{W_kn}\sum_{\beta}\la\Gamma^\beta\ra D_{W_k}^{\beta S}D_n^{\beta I}\sum_\alpha D_{W_k}^{\alpha I}\zeta_{W_k}^{\alpha\beta}\prod_{j=1, j\neq k}^rD_{W_j}^{\alpha A_j}D_{W_j}^{\beta A_j} \\ \nonumber
&&+\sum_{k=1}^rD_k^{AR}\gamm_{W_k}\sum_\beta \la\Gamma^{\beta}\ra D_{W_k}^{\beta I}\sum_\alpha D_{W_k}^{\alpha R}\zeta_{W_k}^{\alpha\beta}\prod_{j=1,j\neq k}^rD_{W_j}^{\alpha A_j}D_{W_j}^{\beta A_j} \\ \nonumber
&&-\sum_{k=1}^rD_k^{AS}\sum_{n=1}^PT_{W_kn}\sum_\alpha\la\Gamma^{\alpha}\ra D_{W_k}^{\alpha S}D_n^{\alpha I}\sum_\beta  D_{W_k}^{\beta I}\zeta_{W_k}^{\beta\alpha} \prod_{j=1,j\neq k}^rD_{W_j}^{\alpha A_j}D_{W_j}^{\beta A_j}\\ \nonumber
&&-\sum_{k=1}^rD_k^{AI}\gamm_{W_k}\sum_\alpha \la\Gamma^{\alpha}\ra D_{W_k}^{\alpha I}\sum_\beta D_{W_k}^{\beta R}\zeta_{W_k}^{\beta\alpha}\prod_{j=1,j\neq k}^rD_{W_j}^{\alpha A_j}D_{W_j}^{\beta A_j}.
\end{eqnarray}
Applying Proposition~\ref{prop6} gives
\begin{eqnarray}\nonumber
\dot{\la\psi_W^A\ra}&=&\sum_{k=1}^rD_k^{AI}\sum_{n=1}^PT_{W_kn}\sum_{\beta}\la\Gamma^\beta\ra D_{W_k}^{\beta S}D_n^{\beta I}\prod_{j=1, j\neq k}^rD_{W_j}^{\beta A_j} \\ \nonumber
&&+\sum_{k=1}^rD_k^{AR}\gamm_{W_k}\sum_\beta \la\Gamma^{\beta}\ra D_{W_k}^{\beta I}\prod_{j=1,j\neq k}^rD_{W_j}^{\beta A_j} \\ \nonumber
&&-\sum_{k=1}^rD_k^{AS}\sum_{n=1}^PT_{W_kn}\sum_\alpha\la\Gamma^{\alpha}\ra D_{W_k}^{\alpha S}D_n^{\alpha I}\prod_{j=1,j\neq k}^rD_{W_j}^{\alpha A_j}\\ \nonumber
&&-\sum_{k=1}^rD_k^{AI}\gamm_{W_k}\sum_\alpha \la\Gamma^{\alpha}\ra D_{W_k}^{\alpha I}\prod_{j=1,j\neq k}^rD_{W_j}^{\alpha A_j}.
\end{eqnarray}
Breaking up the sums over $n$ on the first and third lines depending on whether the node $n$ is internal or external to the motif $\psi_W$ gives:
\begin{eqnarray}\nonumber
\dot{\la\psi_W^A\ra}&=&\sum_{k=1}^rD_k^{AI}\sum_{n\notin W}T_{W_kn}\sum_{\beta}\la\Gamma^\beta\ra D_{W_k}^{\beta S}D_n^{\beta I}\prod_{j=1, j\neq k}^rD_{W_j}^{\beta A_j} \\ \nonumber
&&+\sum_{k=1}^rD_k^{AI}\sum_{n\in W}T_{W_kn}\sum_{\beta}\la\Gamma^\beta\ra D_{W_k}^{\beta S}D_n^{\beta I}\prod_{j=1, j\neq k}^rD_{W_j}^{\beta A_j} \\ \nonumber
&&+\sum_{k=1}^rD_k^{AR}\gamm_{W_k}\sum_\beta \la\Gamma^{\beta}\ra D_{W_k}^{\beta I}\prod_{j=1,j\neq k}^rD_{W_j}^{\beta A_j} \\ \nonumber
&&-\sum_{k=1}^rD_k^{AS}\sum_{n\notin W}T_{W_kn}\sum_\alpha\la\Gamma^{\alpha}\ra D_{W_k}^{\alpha S}D_n^{\alpha I}\prod_{j=1,j\neq k}^rD_{W_j}^{\alpha A_j} \\ \nonumber
&&-\sum_{k=1}^rD_k^{AS}\sum_{n\in W}T_{W_kn}\sum_\alpha\la\Gamma^{\alpha}\ra D_{W_k}^{\alpha S}D_n^{\alpha I}\prod_{j=1,i\neq k}^rD_{W_j}^{\alpha A_j} \\ \nonumber
&&-\sum_{k=1}^rD_k^{AI}\gamm_{W_k}\sum_\alpha \la\Gamma^{\alpha}\ra D_{W_k}^{\alpha I}\prod_{j=1,j\neq k}^rD_{W_j}^{\alpha A_j}.
\end{eqnarray}

Lines 1 and 4 can be immediately recognised as the generating rule (Definition~\ref{GR}). For $n\notin W$ and $n\in N_{W_k}$:
\begin{eqnarray}\nonumber
D_k^{AI}\sum_\beta\la\Gamma^\beta\ra D_{W_k}^{\beta S}D_n^{\beta I}\prod_{j=1,j\neq k}^rD_{W_j}^{\beta A_j}&=&D_k^{AI}\la g_{W_k}^n(\psi_W^A)\ra, \\ \nonumber
D_k^{AS}\sum_\alpha\la\Gamma^\alpha\ra D_{W_k}^{\alpha S}D_n^{\alpha I}\prod_{j=1,j\neq k}^rD_{W_j}^{\alpha A_j}&=&D_k^{AS}\la g_{W_k}^n(\psi_W^A)\ra,
\end{eqnarray}
and if $n\notin N_{W_k}$ then $T_{W_kn}=0$.

Line 2 requires that $n\in W$. Let $l\in \{1,2,...,r\}$ and $W_l=n$. Then:
\begin{eqnarray}\nonumber
&&\sum_{l=1}^rT_{W_kW_l}\sum_\beta\la\Gamma^\beta\ra D_{W_k}^{\beta S}D_n^{\beta I}\prod_{j=1,j\neq k}^rD_{W_j}^{\beta A_j} \\ \nonumber
&&=\sum_{l=1}^rT_{W_kW_l}\sum_\beta\la\Gamma^\beta\ra D_{W_k}^{\beta S}D_{W_l}^{\beta I}D_{W_l}^{\beta A_l}\prod_{j=1,j\neq k, j\neq l}^rD_{W_j}^{\beta A_j} \\ \nonumber
&&=\sum_{l=1}^rT_{W_kW_l}D_l^{AI}\sum_\beta\la\Gamma^\beta\ra D_{W_k}^{\beta S}D_{W_l}^{\beta I}\prod_{j=1,j\neq k, j\neq l}^rD_{W_j}^{\beta A_j},
\end{eqnarray}
where the last equality follows from Proposition~\ref{prop5}. Using the definition of $h_{W_k}(\psi_W^A)$, this becomes:
\begin{eqnarray*}
\sum_{l=1}^rT_{W_kW_l}D_l^{AI}\la h_{W_k}(\psi_W^A)\ra,
\end{eqnarray*}
and similarly for line 5.

We obtain:
\begin{eqnarray}\nonumber
\dot{\la\psi_W^A\ra}&=&\sum_{k=1}^rD_k^{AI}\sum_{n\notin W}T_{W_kn}\la g_{W_k}^n(\psi_W^A)\ra \\ \nonumber
&&+\sum_{k=1}^rD_k^{AI}\sum_{l=1}^rD_l^{AI}T_{W_kW_l}\la h_{W_k}(\psi_W^A)\ra \\ \nonumber
&&+\sum_{k=1}^rD_k^{AR}\gamm_{W_k}\la f_{W_k}(\psi_W^A)\ra \\ \nonumber
&&-\sum_{k=1}^rD_k^{AS}\sum_{n\notin W}T_{W_kn}\la g_{W_k}^n(\psi_W^A)\ra \\ \nonumber
&&-\sum_{k=1}^rD_k^{AS}\sum_{l=1}^rD_l^{AI}T_{W_kW_l}\la h_{W_k}(\psi_W^A)\ra \\ \nonumber
&&-\sum_{k=1}^rD_k^{AI}\gamm_{W_k}\la\psi_W^A\ra,
\end{eqnarray}
where the $h_{W_k}$ operator on the 5th line is superfluous but allows us to write the equation in the form of (\ref{14}).
\end{proof}


\section{Acknowledgements}
This research was facilitated in part by the Research Centre for Mathematics and Modelling at The University of Liverpool. We thank two anonymous reviewers for helpful comments which improved the manuscript.

\end{document}